\documentclass[12pt]{article}

\usepackage{amsthm}

\usepackage[a4paper, includefoot,
            left=3cm, right=1.5cm,
            top=2cm, bottom=2cm,
            headsep=1cm, footskip=1cm]{geometry}

\theoremstyle{definition}
\newtheorem{lemma}{Lemma}[section]
\newtheorem{proposition}{Proposition}[section]
\newtheorem{remark}{Remark}[section]

\usepackage[utf8]{inputenc}
\usepackage[T2A]{fontenc}
\usepackage[english]{babel}
\usepackage[intlimits]{amsmath}
\usepackage{amssymb}
\usepackage{subfigure}
\usepackage{authblk}
\allowdisplaybreaks[4]

\usepackage{graphicx}

\newtheorem{definition}{Definition}

\newtheorem{theorem}{Theorem}[section]

\newtheorem{corollary}{Corollary}[section]

\newcommand{\footremember}[2]

\begin{document}

\title{Two asymptotic approaches for the exponential signal and harmonic noise    in Singular Spectrum Analysis}
\author[1]{Elizaveta Ivanova}
\author[2]{Vladimir Nekrutkin}
\affil[1]{Junior Researcher, Speech Technology Center, 4 Krasutsky str., St. Petersburg, 196084, Russia,
E-mail imst563@mail.com}
\affil[2]{Associate professor, St. Petersburg State University, 7/9 Universitetskaya  nab., St. Petersburg, 199034, Russia,
E-mail vnekr@statmod.ru}
\date{}
\maketitle

\begin{abstract}
The general theoretical approach to the asymptotic extraction of the signal series from the perturbed signal with the help of Singular Spectrum Analysis (briefly, SSA) was already outlined in Nekrutkin 2010, SII, v. 3, 297--319.

In this paper we consider the example of such an analysis applied to the increasing exponential signal and the sinusoidal noise.
It is proved that if the signal rapidly tends to infinity, then the so-called reconstruction errors of SSA do not uniformly tend to  zero as the series length tends to infinity. More precisely, in this case any finite number of last terms of the error series do not tend to any
finite or infinite values.

On the contrary, for the ``discretization'' scheme with the
 bounded from above  exponential signal, all elements of the error series tend to zero.
 This effect shows that the discretization model can be an effective tool in the theoretical SSA considerations with increasing signals.

\end{abstract}



\section{Introduction}
Let us start with the general construction described in \cite{Nekrutkin10}.
Consider the real-valued ``signal'' series $\mathrm{F}_N =(x_1,\ldots,x_{N-1})$, $1<L<N-1$.
Transfer the series $\mathrm{F}_N$ into the  Hankel ``trajectory''  $L\times K$-matrix ${\bf H}$
with rows $(x_j,\ldots,x_{K+j-1})$, where $0\leq j<L$ and $L+K=N+1$.

It is supposed that $d\stackrel{\rm def}={\rm rank}\, {\bf H}< \min(K,L)$.
Denote $\mathbb{U}_0$ the eigenspace corresponding to the zero eigenvalue of the matrix
 ${\bf A}\stackrel{\rm def}={\bf H}{\bf H}^{\rm T}$.
Then $ d=\dim \mathbb{U}_0^\perp$ and $\dim \mathbb{U}_0=K-d>0$.

Let $\mathrm{F}_N(\delta)\!=\!\mathrm{F}_N\!+\!\delta\mathrm{E}_N$  be the perturbed signal, where $\mathrm{E}_N=(e_0,\ldots,e_{N-1})$
is a certain ``noise'' series and  $\delta$ stands for a formal perturbation parameter.
Then we come to the perturbed matrix ${\bf H}(\delta)\!=\!{\bf H}\!+\!\delta {\bf E}$ with the Hankel matrix ${\bf E}$ produced from the
noise series~$\mathrm{E}_N$.

If  $\delta$ is sufficiently small, then the linear space $\mathbb{U}_0^\perp(\delta)$ spanned by
$d$ main left singular vectors of the matrix ${\bf H(\delta)}$ can serve as an  approximation to
 $\mathbb{U}_0^\perp$.
 The quality of this approximation can be measured by the spectral norm  $\big\|{\bf P}_0^\perp(\delta)-{\bf P}_0^\perp\big\|$, where
 ${\bf P}_0^\perp$ and ${\bf P}_0^\perp(\delta)$ are orthogonal projections on the linear spaces $\mathbb{U}_0^\perp$ and $\mathbb{U}_0^\perp(\delta)$
   correspondingly. Note that $\big\|{\bf P}_0^\perp(\delta)-{\bf P}_0^\perp\big\|$ is nothing but the sine
of   the largest principal angle between unperturbed and perturbed signal subspaces $\mathbb{U}_0^\perp$ and $\mathbb{U}_0^\perp(\delta)$.

It is well-known that a lot of subspace-based methods of signal processing are relying on the close proximity of $\mathbb{U}_0^\perp$ and $\mathbb{U}_0^\perp(\delta)$. Still  the main goal of  Singular Spectrum Analysis (briefly, SSA) is the approximation (or ``reconstruction'') of the signal
$\mathrm{F}_N$ from the perturbed signal $\mathrm{F}_N(\delta)$, see \cite{GNZh01} for the detailed description.

As it is mentioned in \cite[sect. 5]{Nekrutkin10}, the analysis of the errors of this approximation can be expressed   in such a manner.
First of all,  the ``hankelization'' (in other terms, ``diagonal averaging'') operator $\mathcal{S}$ is defined.

If the hankelization operator $\mathcal{S}$ is applied to some $L\times K$ matrix $\mathbf{Y}=\{y_{k,\ell}\}_{k=1,\ell=1}^{L,K}$ then the resulting $L\times K$ matrix  $\mathcal{S}\mathbf{Y}$
has equal values denoted by $(\mathcal{S}\mathbf{Y})_j$ on its anti-diagonals $\{(k,\ell): \text{such that } k+\ell-2=j$\}, where $j=0,\ldots, N-1$, $k=1,\ldots, L$ and
$\ell=1,\ldots, K$. Besides, $(\mathcal{S}\mathbf{Y})_j$ equals to the average of inputs $y_{k,\ell}$ on this anti-diagonal.

Then, under denotation
\begin{gather}
\label{eq:deltaH_pres_1}
\Delta_{\delta}({\bf H}) =
({\bf P}_0^\perp(\delta) - {\bf P}_0^\perp){\bf H}(\delta) + \delta {\bf P}_0^\perp{\bf E},
\end{gather}
 the series $r_0,\ldots, r_{N-1}$ with
 \begin{gather}
 \label{eq:rec_err}
 r_j=\big(\mathcal{S}\Delta_{\delta}({\bf H})\big)_j
 \end{gather}
  is the series of the reconstruction SSA errors.

The peculiarity of the approach described in \cite{Nekrutkin10} can be  expressed as follows. The usual method of the general
perturbation  analysis is to consider the small perturbation parameter $\delta$ and therefore to investigate the linear in $\delta$
approximation of the problem. For the subspace-based methods this means
 that $N$ is fixed and $\delta \downarrow 0$, see for example
\cite{BadeauRD_08} - \cite{VSch17}.

Still  SSA is usually characterized by big series length $N$, and formally this corresponds to fixed $\delta$ and $N\rightarrow \infty$.
The mathematical technique that is used in \cite{Nekrutkin10} for this goal, goes back to \cite{Kato66}
and consists of the asymptotic analysis of the corresponding perturbation expansions.

  This paper is devoted to the example of such an analysis for the exponentially growing signal and the harmonic noise.
  This model is not so far from real-life series.  For example, the series ``Gasoline demand'' (see Fig.  \ref{fig:gasoline}, data is taken from \cite{AbrahamL05})
  can be approximated by the sum of two addends: the increasing trend of the exponential form and the 12-month  periodicity.

 \begin{figure}[h!]
\centering
\includegraphics[width=.7\linewidth]{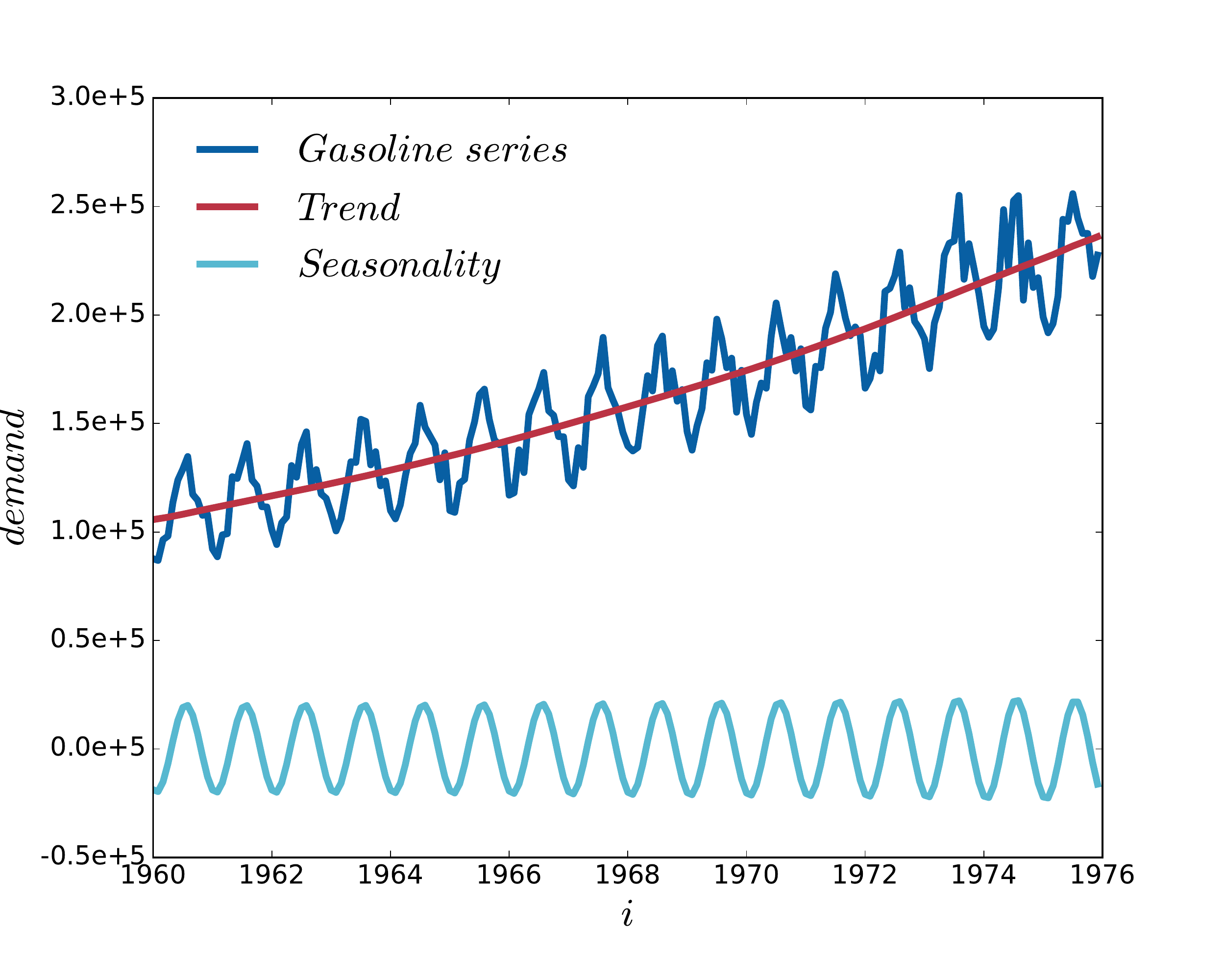}
\caption{Gasoline demand, monthly Jan 1960 -- Dec 1975, gallon millions, Ontario.}
\label{fig:gasoline}
\end{figure}

  Note that both trend and periodicity in Fig. \ref{fig:gasoline} are produced by SSA with $L=N/2=96$. Naturally, the trend is reconstructed by the first eigentriple of the decomposition, while the periodicity is produced with the help of eigentriples 2 and 3.


In this paper we deal with  the following construction.
A certain time period $[0,T]$ is divided into $N$ intervals of length $\Delta=T/N$, and we consider the signal $x_n=e^{\theta \Delta n}$
and the noise $e_n=\cos(\xi n+\varphi)$, so that the perturbed signal has the form
\begin{gather}
\label{eq:series_gen}
  f_n=e^{\theta \Delta n}+\delta \cos(\xi n+\varphi),\quad n=0,\ldots, N-1,
\end{gather}
where  $\theta >0$, $\xi=2\pi \omega$ with $\omega \in (0,1/2)$, and $\varphi \in [0,2\pi)$.

As in \cite{Nekrutkin10}, we are interesting in the behavior of the reconstruction errors for long signals. For this goal we consider
two asymptotic schemes as $N\rightarrow \infty$.
\begin{enumerate}
  \item
  $\Delta $ is fixed, further we put $\Delta=1$. Then $T=N\rightarrow \infty$ and \eqref{eq:series_gen} has the form
 \begin{gather}
\label{eq:series_incrT}
  f_n=a^{ n}+\delta \cos(\xi n+\varphi),\quad n=0,\ldots, N-1
\end{gather}
with $a=e^\theta>1$.
  \item
  $T$ is fixed  and $\Delta=T/N\rightarrow 0$. Then we come to
  \begin{gather}
\label{eq:series_constT}
  f_n=a^{T n/N}+\delta \cos(\xi n+\varphi),\quad n=0,\ldots N-1
\end{gather}
with the same $a$. Further we apply the term ``discretization'' for this scheme.
\end{enumerate}

Note that in both cases $d=\dim \mathbb{U}_0^\perp=1$ for $L,K>1$. Yet there are considerable
 differences between \eqref{eq:series_incrT} and \eqref{eq:series_constT}. In particular, the signal of the series \eqref{eq:series_incrT} tends to
 infinity as $N\rightarrow \infty$ while $a^{T n/N}<a^T= {\rm const}$ in \eqref{eq:series_constT}.
Though the number of noise periods tends to infinity in both models, the discretization model  seems to
 describe the real-life situations better than \eqref{eq:series_incrT}. For example, the trend of ``Gasoline demand'' series growths very slowly over the period of observations,
while the number of periods of the 12-month harmonic  is relatively~big.

For both models, our interest lies in the asymptotic behavior of $\big\|{\bf P}_0^\perp(\delta)-{\bf P}_0^\perp\big\|$ and of the reconstruction errors. Since the model \eqref{eq:series_incrT} corresponds to the style of all examples in~\cite{Nekrutkin10},
several results about this series can be borrowed from \cite{Nekrutkin10}.

In particular, it is already shown for \eqref{eq:series_incrT}, see \cite[sect. 3.2.1]{Nekrutkin10}, that under the conditions $N\rightarrow \infty$ and $\min(L,K)\rightarrow \infty$
\begin{gather}
\label{eq:diff_proj}
\big\|{\bf P}_0^\perp(\delta)-{\bf P}_0^\perp\big\|=O(N a^{-N})
\end{gather}
and
\begin{gather}
\label{eq:main_term V}
\big\|{\bf P}_0^{\perp}(\delta)-{\bf P}_0^{\perp}-\delta{\bf V}_0^{(1)}\big\|=O(N^2 a^{-2N})
\end{gather}
for any $\delta\in \mathbb{R}$, where
\begin{gather}
{\bf V}^{(1)}_0=
{\bf P}_{0}{\bf E}{\bf H}^{\rm T}{\bf S}_0+{\bf S}_0{\bf H}{\bf E}^{\rm T}{\bf P}_{0}
\label{eq:lin_proj}
\end{gather}
is the linear  term of the  expansion of ${\bf P}_0^{\perp}(\delta)-{\bf P}_0^{\perp}$ into power series (see
\cite[theor. 2.1]{Nekrutkin10}), and $\mathbf{S}_0$ stands for the pseudoinverse of $\mathbf{HH}^{\rm T}$. Besides,
$\big\|{\bf V}^{(1)}_0\| = O \big(N a^{-N}\big)$.

Note that in the case $L\sim \alpha N$ with $\alpha \in (0,1)$ more careful  calculations lead to the precise asymptotic
\begin{gather}
\label{eq:diffP_pr}
\frac{a^{N}}{\sqrt{N}}\ \big\|{\bf P}_0^{\perp}(\delta) - {\bf P}_0^{\perp}\big\| \rightarrow |\delta|\,\frac{a^2 - 1}{a}\ \sqrt{\frac{\alpha(a^2 - 1)}{2(a^2 + 1 - 2a\cos\xi)}}
\end{gather}
as well as to more precise inequality
\begin{gather}
\label{eq:main_term V_pr}
\big\|{\bf P}_0^{\perp}(\delta)-{\bf P}_0^{\perp}-\delta{\bf V}_0^{(1)}\big\|=O(N^{3/2}a^{-2N})
\end{gather}
instead of \eqref{eq:main_term V}.

Since
we omit here proofs of both   \eqref{eq:diffP_pr} and \eqref{eq:main_term V_pr}, we use the inequality \eqref{eq:main_term V}
for the series \eqref{eq:series_incrT} in all
further considerations.


Section \ref{sect:rec_an} of the paper is devoted to the reconstruction errors $r_j=r_j(N)$  for the model \eqref{eq:series_incrT}.
Proposition \ref{prop:rho_rec_cos_noise} and Corollary \ref{cor:err_rec_cos_noise} show that  $r_j\rightarrow 0$ as
$N\rightarrow \infty$ if, roughly speaking, $j$ is separated from $N$.

On the contrary, if $j$ is close to $N$, then $r_j$ does not converge to zero. Moreover, the  asymptotic  behavior of $r_j$ in this case
depends on  the rationality/irrationality of the frequency $\omega=\xi/2\pi$, see propositions \ref{prop:limpoints_rat} and
\ref{prop:r_irrat}.

The model \eqref{eq:series_constT} under the assumption $L/N\rightarrow \alpha \in (0,1) $ is investigated in  Section \ref{sect:rec_ant}. It is proved in Theorem \ref{theor:diff_proj_discr} that in this case $\big\|\mathbf{P}_0^\perp(\delta)-\mathbf{P}_0^\perp\big\|=O(N^{-1})$
as $N\rightarrow \infty$ for the sufficiently small $\delta$.

Unlike the model \eqref{eq:series_incrT}, the reconstruction errors $r_j$ in the discretization scheme tend to zero for all $j$,
see Theorem \ref{theor:main_discr}. Thus the model \eqref{eq:series_constT} seems to be more practical than   \eqref{eq:series_incrT}.

In what follows, we always assume the {\it regular behavior} of the parameter $L=L(N)$ as $N\rightarrow \infty$. This means that
$L\sim \alpha N$ with $\alpha\in (0,1)$. Still several further inequalities  are valid under less restrictive condition
$\min(L,K)\rightarrow \infty$.

\section{Reconstruction errors for the signal $x_n=a^n$}
\label{sect:rec_an}

Consider the series \eqref{eq:series_incrT} and
suppose that  $L\sim \alpha N$ with $\alpha\in (0,1)$ as $N\rightarrow \infty$.
Our aim is to investigate the asymptotic properties of the reconstruction errors
\eqref{eq:deltaH_pres_1}, \eqref{eq:rec_err} for the perturbed series \eqref{eq:series_incrT}.
Since the result of the reconstruction does not
change if we reduce $\mathbf{H}$ by $\mathbf{H}^{\rm T}$, assume that
$L\leq K$.

The base of the approach is the well-known inequality $\|\mathbf{A}\|_{\max}\leq \|\mathbf{A}\|$, where $\|\mathbf{A}\|$ stands for the usual spectral norm of the
matrix $\mathbf{A}$ and  $\|\mathbf{A}\|_{\max}=\max |a_{ij}|$ for the matrix $\mathbf{A}$ with entries $a_{ij}$. Therefore, if $\|\mathbf{A}\|$ is small, then
$\|\mathcal{S}\mathbf{A}\|_{\max}$ is small as well.

Thus we rewrite \eqref{eq:deltaH_pres_1} in the form
\begin{gather}
\label{eq:deltaH_pres1}
\Delta_{\delta}({\bf H}) =
\big({\bf P}_0^\perp(\delta) - {\bf P}_0^\perp- \delta{\bf V}_0^{(1)}\big){\bf H}(\delta) + \delta {\bf P}_0^\perp{\bf E}+
\delta{\bf V}_0^{(1)}\big({\bf H}+\delta\mathbf{E}\big)\ .
\end{gather}

It is easy to check that  $\|{\bf E}\|=O(N)$, $\|{\bf H}\|=O(a^N)$ and $\|{\bf V}_0^{(1)}{\bf E}\| = O(N^{2}a^{-N})$.
Applying \eqref{eq:main_term V}, we see that
\begin{gather*}
\|\big({\bf P}_0^{\perp}(\delta) - {\bf P}_0^{\perp} - \delta{\bf V}_0^{(1)}\big){\bf H}(\delta)\| = O(N^{2}a^{-N}).
\end{gather*}

 This means that the reconstruction errors have the form
 \begin{gather}
 \label{eq:rj_gen}
r_j=r_j(N)= \delta\big(\mathcal{S}({\bf V}_0^{(1)}{\bf H}+{\bf P}_0^\perp{\bf E})\big)_j+O(N^2 a^{-N}), \quad j=0,\ldots, N-1,
 \end{gather}
 and all we need is to investigate the asymptotical behavior of the  series
   \begin{gather}
   \label{eq:rho_j}
   \rho_j\stackrel{\rm def}=\big(\mathcal{S}({\bf V}_0^{(1)}{\bf H}+{\bf P}_0^\perp{\bf E})\big)_j=
   \big(\mathcal{S}({\bf P}_{0}{\bf E}{\bf H}^{\rm T}{\bf S}_0\mathbf{H}\big)_j+
  \big(\mathcal{S}({\bf P}_0^\perp{\bf E})\big)_j.
   \end{gather}

   \subsection{The reconstruction errors $r_j(N)$ in the case $N-j\rightarrow \infty$}
   Let us start with the case when $j$ is not close to $N$.
  \begin{proposition}
\label{prop:rho_rec_cos_noise}
Let $\rho_j=\rho_j(N)$ be defined by  \eqref{eq:rho_j}.
If $L$ is regular as $N\rightarrow   \infty$  and $N-j\rightarrow \infty$, then $\rho_j(N)\rightarrow 0$ as $N\rightarrow \infty$.
 \end{proposition}
 \begin{proof}
 As it was already mentioned, it is sufficient to assume that $L\leq K$.
 First of all, for fixed
$\xi\in (0,\, \pi)$, $b > 1$, $\psi \in [0,\, 2\pi)$ and an integer
$M \geq 1$ denote
\begin{gather}
\label{eq:phi_l_eq}
\varPhi_M(b, \psi) = \sum_{j = 0}^{M - 1}b^j\cos(\xi j + \psi)
\end{gather}
and
\begin{gather}
\label{eq:upsilon}
\varUpsilon_{T, M}(b, \psi) = \sum_{j = 0}^{T - 1}b^j\varPhi_M(b, \xi j + \psi).
\end{gather}
Evidently,
\begin{gather}
\label{eq:varPhi_varUpsilon_brief}
|\varPhi_M(b, \psi)|\leq (b^M-1)/(b-1) \quad \text{and} \quad
|\varUpsilon_{T, M}(b, \psi)|\leq \frac{(b^M-1)(b^T-1)}{(b-1)^2}\ .
\end{gather}

Under the denotation $W_M=(0,a,\ldots, a^{M-1})^{\rm T}$,
\begin{gather}
\label{eq:P0S0}
{\bf P}_0^{\perp} = \frac{W_LW_L^{\rm T}}{\|W_L\|^2} \quad \text{and} \quad
{\bf S}_0 = \frac{W_LW_L^{\rm T}}{\|W_L\|^4\|W_K\|^2}
\end{gather}
and therefore
\begin{gather*}
{\bf P}_0^{\perp} {\bf E} + {\bf V}_0^{(1)}{\bf H} = \frac{W_LW_L^{\rm T}}{\|W_L\|^2}\,{\bf E} + \left({\bf I} - \frac{W_LW_L^{\rm T}}{\|W_L\|^2}\right){\bf E}{\bf H}^{\rm T}\,\frac{W_LW_L^{\rm T}}{\|W_L\|^4\|W_K\|^2}\,{\bf H} =\\
\frac{W_LW_L^{\rm T}{\bf E}}{\|W_L\|^2} + \frac{{\bf E}W_KW_K^{\rm T}}{\|W_K\|^2} -
\varUpsilon_{L,K}(a, \varphi)\,\frac{W_LW_K^{\rm T}}{\|W_L\|^2\|W_K\|^2}= \mathbf{J}_1+\mathbf{J}_2+\mathbf{J}_3.
\end{gather*}

Since $\big(\mathcal{S}(W_LW_K^{\rm T})\big)_j=a^j$ and $\|W_M\|^2=(a^{2M}-1)/(a^2-1)$, then
\begin{gather*}
\left|\big(\mathcal{S}(\mathbf{J}_3)\big)_j\right|= |\varUpsilon_{L,K}(a, \varphi)|\ \frac{a^j(a^2-1)^2}{(a^{2L}-1)(a^{2K}-1)}\leq
\\
\frac{(a^L-1)(a^K-1)}{(a-1)^2}\ \frac{a^j(a^2-1)^2}{(a^{2L}-1)(a^{2K}-1)}=
\frac{a^j(a+1)^2}{(a^{L}+1)(a^{K}+1)}\ .
\end{gather*}

Let us now check $\mathbf{J}_1$ and $\mathbf{J}_2$.
In view of the equalities
\begin{gather}
\label{tq:EWK}
{\bf E}W_K=
\big(\varPhi_K(a, \varphi), \varPhi_K(a, \xi + \varphi), \ldots, \varPhi_K(a, (L - 1)\xi + \varphi)\big)^{\rm T}
\end{gather}
and
\begin{gather*}
{\bf E}^{\rm T}W_L =
\big(\varPhi_L(a, \varphi), \varPhi_L(a, \xi + \varphi), \ldots, \varPhi_L(a, (K - 1)\xi + \varphi)\big)^{\rm T},
\end{gather*}
we get that
\begin{gather}
\big(\mathcal{S}(W_LW_L^{\rm T}{\bf E})\big)_j
=
\begin{cases}
\displaystyle
      \frac{1}{j + 1} \sum_{k = 0}^ja^{k}\varPhi_L(a, (j - k)\xi + \varphi) &\text{for } 0 \leq j < L,\\
\displaystyle
      \frac{1}{L} \sum_{k = 0}^{L - 1}a^{k}\varPhi_L(a, (j -k)\xi + \varphi) &\text{for } L \leq j < K,\\
\displaystyle
      \frac{1}{N \!-\! j} \sum_{k = j - K + 1}^{N - K}a^{k}\varPhi_L(a, (j - k)\xi + \varphi) &\text{for } K \leq j < N.
\end{cases}
\label{eq:SWWE}
\end{gather}
In the same manner,
\begin{gather*}
\big(\mathcal{S}({\bf E}W_KW_K^{\rm T})\big)_j
=
\begin{cases}
\displaystyle
      \frac{1}{j + 1} \sum_{k = 0}^ja^{j - k}\varPhi_K(a, k\xi + \varphi) &\text{for } 0 \leq j < L,\\
\displaystyle
      \frac{1}{L} \sum_{k = 0}^{L - 1}a^{j - k}\varPhi_K(a, k\xi + \varphi) &\text{for } L \leq j < K,\\
\displaystyle
      \frac{1}{N - j} \sum_{k = j - K + 1}^{N - K}a^{j - k}\varPhi_K(a, k\xi + \varphi) &\text{for } K \leq j < N.
\end{cases}
\end{gather*}
Due to
\eqref{eq:varPhi_varUpsilon_brief},
\begin{gather*}
\Big|\big(\mathcal{S}(W_LW_L^{\rm T}{\bf E})\big)_j\Big|
\leq
\begin{cases}
\displaystyle
      \frac{1}{j + 1} \frac{a^{j+1}-1}{a-1}
      \frac{a^L-1}{a-1}
       &\text{for } 0 \leq j < L,\\
\displaystyle
      \frac{1}{L}
      \left(\frac{a^L-1}{a-1}\right)^2
       &\text{for } L \leq j < K,\\
\displaystyle
      \frac{a^{j-K+1}}{N \!-\! j}\  \frac{a^{N-j}-1}{a-1}\
      \frac{a^L-1}{a-1}
       &\text{for } K \leq j < N.
\end{cases}
\end{gather*}
and
\begin{gather*}
\Big|\big(\mathcal{S}({\bf E}W_KW_K^{\rm T})\big)_j\Big|
\leq
\begin{cases}
\displaystyle
      \frac{1}{j + 1} \frac{a^{j+1}-1}{a-1}
           \frac{a^K-1}{a-1}
      &\text{for } 0 \leq j < L,\\
\displaystyle
      \frac{a^j}{L}\ \frac{a^L-1}{a^L}\frac{1}{a-1}
      \frac{a^K-1}{a-1}
       &\text{for } L \leq j < K,\\
\displaystyle
      \frac{ a^{j-N+K}}{N - j}\ \frac{a^{N-j}-1}{a-1}\,
      \frac{a^K-1}{a-1}
       &\text{for } K \leq j < N.
\end{cases}
\end{gather*}

Therefore,
\begin{gather*}
  |\rho_j| \leq
 \frac{1}{j + 1} \frac{a+1}{a-1} \frac{a^{j+1}-1}{a^{L}+1} +\frac{1}{j + 1} \frac{a+1}{a-1} \frac{a^{j+1}-1}{a^{K}+1}   +
         \frac{a^j(a+1)^2}{(a^{L}+1)(a^{K}+1)}
\end{gather*}
for  $0 \leq j < L$,
\begin{gather*}
|\rho_j|\leq
      \frac{1}{L}\, \frac{a+1}{a-1}\frac{a^L-1}{a^L+1}+\frac{1}{L}\,\frac{a+1}{a-1}\, \frac{a^L-1}{a^L}\frac{a^j}{a^K+1}
      +\frac{a^j(a+1)^2}{(a^{L}+1)(a^{K}+1)}
      \end{gather*}
for $L \leq j < K$, and
\begin{gather*}
|\rho_j| \leq
\frac{a+1}{a-1}\frac{1}{N-j}\frac{a^L-a^{j-N+L}}{a^L+1}+\frac{a+1}{a-1}\frac{1}{N-j}\frac{a^K-a^{j-N+K}}{a^K+1}+
\frac{a^j(a+1)^2}{(a^{L}+1)(a^{K}+1)}
\end{gather*}
if $K \leq j < N$.
Thus there exist a constant $C$ such that for $N$ big enough
\begin{gather}
|\rho_j|\leq C
\begin{cases}
\displaystyle
{a^{-(L-j)}}/(j+1) & \text{for} \ \  0 \leq j < L,\\
 1/L &\text{for} \ \  L \leq j < K,\\
 1/(N-j)+a^{-(N-j)}& \text{for} \ \  K \leq j < N.
 \label{eq:est_rho}
\end{cases}
\end{gather}
The proof is complete.
 \end{proof}
\begin{corollary}
\label{cor:err_rec_cos_noise}
It follows from Proposition \ref{prop:rho_rec_cos_noise} that under conditions $L/N\rightarrow \alpha\in (0,1)$
and
$N-j\rightarrow \infty$ the reconstruction errors $r_j$ tend to zero as  $N\rightarrow \infty$. Moreover, if $L\leq K$, then
$|r_j|$ follow the same inequalities \eqref{eq:est_rho} as $|\rho_j|$ up to the multiplicator $|\delta|$.
\end{corollary}

\subsection{The reconstruction errors $r_j(N)$ in the case $N-j=O(1)$}
Consider  now the case $j=N-1-\ell$ with $\ell=O(1)$.
Denote
\begin{gather*}
G(a,\xi)=\frac{a^2 - 1}{a(a^2 + 1 - 2a\cos\xi)}\ .
\end{gather*}
\begin{proposition}
\label{prop:last_indices}
If  $\ell = O(1)$ then under the conditions of  Corollary \ref{cor:err_rec_cos_noise}
\begin{gather}
\label{eq:rn1ell}
r_{N - 1 - \ell} = \delta\,G(a,\xi)\,\Big(C_1(\ell)\cos((N - 1)\xi + \varphi) + C_2(\ell)\sin((N - 1)\xi + \varphi)\Big)+ O(a^{-L}),
\end{gather}
where
\begin{gather}
C_1(\ell) = \frac{2}{1 + \ell}\left(a\cos(\ell\xi) - {a^{-\ell}} {\cos\xi}\right) -
(a^2 + 1 - 2a\cos\xi - 2\sin^2\xi)\,{G(a,\xi)}{a^{-\ell}}\ \ \text{and}
\nonumber
\\
C_2(\ell) = \frac{2}{1 + \ell}\left(a\sin(\ell\xi) + {a^{-\ell}}{\sin\xi}\right) - 2\sin\xi(a - \cos\xi)\ {G(a,\xi)}{a^{-\ell}}\ .
\label{eq:Ciell}
\end{gather}
\end{proposition}
\begin{proof}
For fixed $\xi$ denote $P(a,n,\psi) = a\cos((n - 1)\xi + \psi) - \cos(n\xi + \psi)$.
Straightforward calculations show that
 \begin{gather}
 \label{eq:varPhiM}
\varPhi_M(b, \psi) = \frac{b^{M + 1}\cos((M - 1)\xi + \psi) - b^M \cos(M\xi + \psi) -
b \cos(\xi - \psi) + \cos \psi}{b^2 + 1 - 2b \cos \xi}
\end{gather}
and therefore
\begin{gather}
\varPhi_M(b, \psi)=\frac{b^M P(b,M,\psi)}{b^2 + 1 - 2b\cos\xi}+ O(1)
\label{eq:varPhiL}
\end{gather}
for fixed $b$ and $M\rightarrow \infty$.
%

Denote $\psi_n=(K-\ell - n-1)\xi + \varphi$, then $P(a,L,\psi_n)=P(a,N-n-\ell,\varphi)$ and
\begin{gather*}
\sum_{n = 0}^{\ell}a^{-n} P(a,N-n-\ell,\varphi)=a\cos((N - 1 - \ell)\xi + \varphi) - \cos(N\xi + \varphi)a^{-\ell}.
\end{gather*}

Therefore, taking into account that $K\leq j=N-1-\ell$ and applying \eqref{eq:varPhiL} with $b=a$, $M=L$
and $\psi=\psi_n$,
 we get from \eqref{eq:SWWE}
%
\begin{gather*}
\frac{\big(\mathcal{S}(W_LW_L^{\rm T}{\bf E})\big)_{N - 1 - \ell}}{\|W_L\|^2} =
\frac{1}{\|W_L\|^2}\,\frac{1}{\ell+1} \sum_{k = N-K-\ell}^{N - K}a^{k}\varPhi_L(a, (N-1-\ell - k)\xi + \varphi)=
\\
\sum_{n = 0}^{\ell} a^{N-K-n} \varPhi_L(a, (K-\ell - n-1)\xi + \varphi)\ \frac{a^2-1}{(\ell+1)(a^{2L}+1)}=\\
\frac{a^2-1}{\ell+1}\ \sum_{n = 0}^{\ell} a^{L-1-n}\left(\frac{a^L P(a,L,\psi_n)}{a^2 + 1 - 2a\cos\xi}+ O_L(1)\right)
\frac{1}{a^{2L}-1}=
\\
\frac{a^2-1}{a(\ell+1)(a^2 + 1 - 2a\cos\xi)}\sum_{n = 0}^{\ell}a^{-n} P(a,N-n-\ell,\varphi)+O(a^{-L})=\\
= \frac{G(a, \xi)}{1 + \ell}\bigg(a\cos((N - 1 - \ell)\xi + \varphi) - \cos(N\xi + \varphi)a^{-\ell}\bigg) + O(a^{-L}).
\end{gather*}

In the same manner,
\begin{gather*}
\frac{(\mathcal{S}{\bf E}W_KW_K^{\rm T})_{N - 1 - \ell}}{\|W_K\|^2}
=\frac{G(a, \xi)}{1 + \ell}\bigg(a\cos((N - 1 - \ell)\xi + \varphi) - \cos(N\xi + \varphi)a^{-\ell}\bigg) + O(a^{-K}).
\end{gather*}

Now we pass to  $\varUpsilon_{L,K}(a, \varphi)$ which is defined in \eqref{eq:upsilon}.
It easy to check that
if $b>1$ and $T, M \rightarrow \infty$ then
\begin{gather*}
\varUpsilon_{T, M}(b, \psi) = \frac{b^{S + 1}}{(b^2 + 1 - 2b\cos\xi)^2}\ C(b,S,\psi)+ O(b^{\min\{T, M\}}),
\end{gather*}
where $S = T + M - 1$ and
\begin{gather*}
C(b,S,\psi)=b^2\cos((S - 1)\xi + \psi) - 2b\cos(S\xi + \psi) + \cos((S + 1)\xi + \psi).
\end{gather*}

Therefore,
\begin{gather*}
\frac{a^{N - 1 - \ell}\varUpsilon_{L,K}(a, \varphi)}{\|W_L\|^2\|W_K\|^2} = \frac{a^{2N -\ell} C(a,N,\varphi)}{(a^2 + 1 - 2a\cos\xi)^2}\
\frac{(a^2 - 1)^2}{a^{2N+2}(1-a^{-2L})(1-a^{-2K})} + O(a^{-L}) =\\
=\frac{G^2(a, \xi)}{a^{\ell}}\ C(a,N,\varphi)
+ O(a^{-L}).
\end{gather*}
Since
\begin{gather*}
a\cos((N - 1 - \ell)\xi + \varphi) - \cos(N\xi + \varphi)a^{-\ell}=\\
\big(a\cos(\ell\xi)\!-\! a^{-\ell}\cos\xi\big)\cos((N\!-\!1)\xi\!+\!\varphi)\!+\!\big(a\sin(\ell\xi)\!+
\!a^{-\ell}\sin\xi \big)\sin((N\!-\!1)\xi\!+\!\varphi)
\end{gather*}
and
\begin{gather*}
C(a,N,\varphi)=(a^2 - 2a\cos\xi + \cos(2\xi))\cos((N\!-\!1)\xi\!+\!\varphi) + 2\sin\xi(a - \cos\xi)\sin((N\!-\!1)\xi\!+\!\varphi),
\end{gather*}
we get the result in view of \eqref{eq:rj_gen}.
\end{proof}

To analyze the behavior of the right-hand side of  \eqref{eq:rn1ell}  more precisely,
we need some more considerations.

First of all, it is worth to mention, that  $C_1(\ell)C_2(\ell)\neq 0$ for any fixed $\ell$.
Thus we can rewrite the result of Proposition \ref{prop:last_indices} in the form
\begin{gather}
\label{eq:last_indices1}
r_{N - 1 - \ell} = \delta F_N(\ell) + O(a^{-L}),
\end{gather}
where
\begin{gather}
\label{eq:FN}
F_N(\ell)=D(\ell)\sin\big((N - 1)\xi + \varphi_1(\ell)\big)
\end{gather}
with
\begin{gather*}
D(\ell)=\frac{a^2 - 1}{a(a^2 + 1 - 2a\cos\xi)}\ \sqrt{C_1^2(\ell) + C_2^2(\ell)} \quad \text{and}\\
\varphi_1(\ell) = {\rm arccos}\left(\frac{C_2(\ell)}{\sqrt{C_1^2(\ell) + C_2^2(\ell)}}\right) + \varphi.
\end{gather*}

Remind that $\xi=2\pi \omega$ with
$\omega\in (0,1/2)$. It is natural that the asymptotic behavior of $r_{N - 1 - \ell}$ depends on the properties of
the frequency $\omega$.

Suppose that $\omega = p/q$, where  $p$ and $q$ are coprime natural numbers. For fixed $0\leq k<q$ and $\ell\geq 0$  consider the sequence
$N_m^{(k)}=mq+k+1, \ m\geq 1$.
Since
\begin{gather*}
\sin\big(2\pi(N_m^{(k)} - 1)p/q + \varphi_1(\ell)\big)=\sin\big(2\pi k p/q + \varphi_1(\ell)\big),
\end{gather*}
then
\begin{gather}
\label{eq:lim_r_j}
r_{N_m^{(k)} - 1 - \ell} \rightarrow D(\ell)\sin\big(2\pi kp/q + \varphi_1(\ell)\big) \quad \text{as} \ \   m\rightarrow \infty.
\end{gather}
%

\begin{proposition}
\label{prop:limpoints_rat}
Let the conditions of Proposition \ref{prop:last_indices} be fulfilled. Assume that $\ell$ is fixed and that $\omega$ is a rational number.
Denote $\tau$ the number of limit points  of the series $r_{N-1-\ell}$ as $N\rightarrow \infty$.
Then $\tau\geq 2$.
\end{proposition}
\begin{proof}

Since $D(\ell)>0$, it is sufficient to examine the expressions
$
S(k)
=\sin(2\pi kp/q  + \varphi_1(\ell))
$
with $0\leq k<q$.
If $S(k)=s={\rm const}$ for all $k$, then there exist integers $m_k$ such that
\begin{gather*}
2\pi k p/q + \varphi_1(\ell) = (-1)^{m_k}\arcsin s  + m_k\pi, \quad k=0,\ldots,q-1.
\end{gather*}
Therefore, for $0\leq k<q-1$
\begin{gather*}
2\pi p/q
= \big((-1)^{m_{k+1}} -(-1)^{m_k}\big)\arcsin s + \big(m_{k + 1} - m_{k}\big)\pi=\\
\Delta(s,m_k,m_{k+1})\stackrel{\rm def}=
\begin{cases}
2\arcsin s + (m_{k + 1} - m_{k})\pi &\text{for even}\  m_{k + 1} \text{ and odd}\  m_{k},\\
-2\arcsin s + (m_{k + 1} - m_{k})\pi &\text{for odd}\  m_{k + 1} \text{ and even}\ m_{k},\\
(m_{k + 1} - m_{k})\pi, &\text{if}\ m_{k + 1}, m_{k}\ \text{are both odd or even}.
\end{cases}
\end{gather*}
Since $0<2\pi p/q<\pi$, then we immediately come to the inequality $s\neq 0$. Suppose now that $\arcsin s>0$. (The case of $\arcsin s<0$ can be
treated in the same manner.) Then $\Delta(s,m_k,m_{k+1})\in (0,\pi)$ iff $m_{k + 1} - m_{k}=1$ and $m_{k+1}$ is odd.

Therefore, for any sequence   $\{m_k\}_{k = 0}^{q - 1}$ with $q > 2$ there exist a pair $(m_{k+1},m_k)$ with $\Delta(s,m_k,m_{k+1})\notin (0,\pi)$, and the
assertion is proved.
\end{proof}

\begin{figure}[h!]
\centering
\includegraphics[width=.7\linewidth]{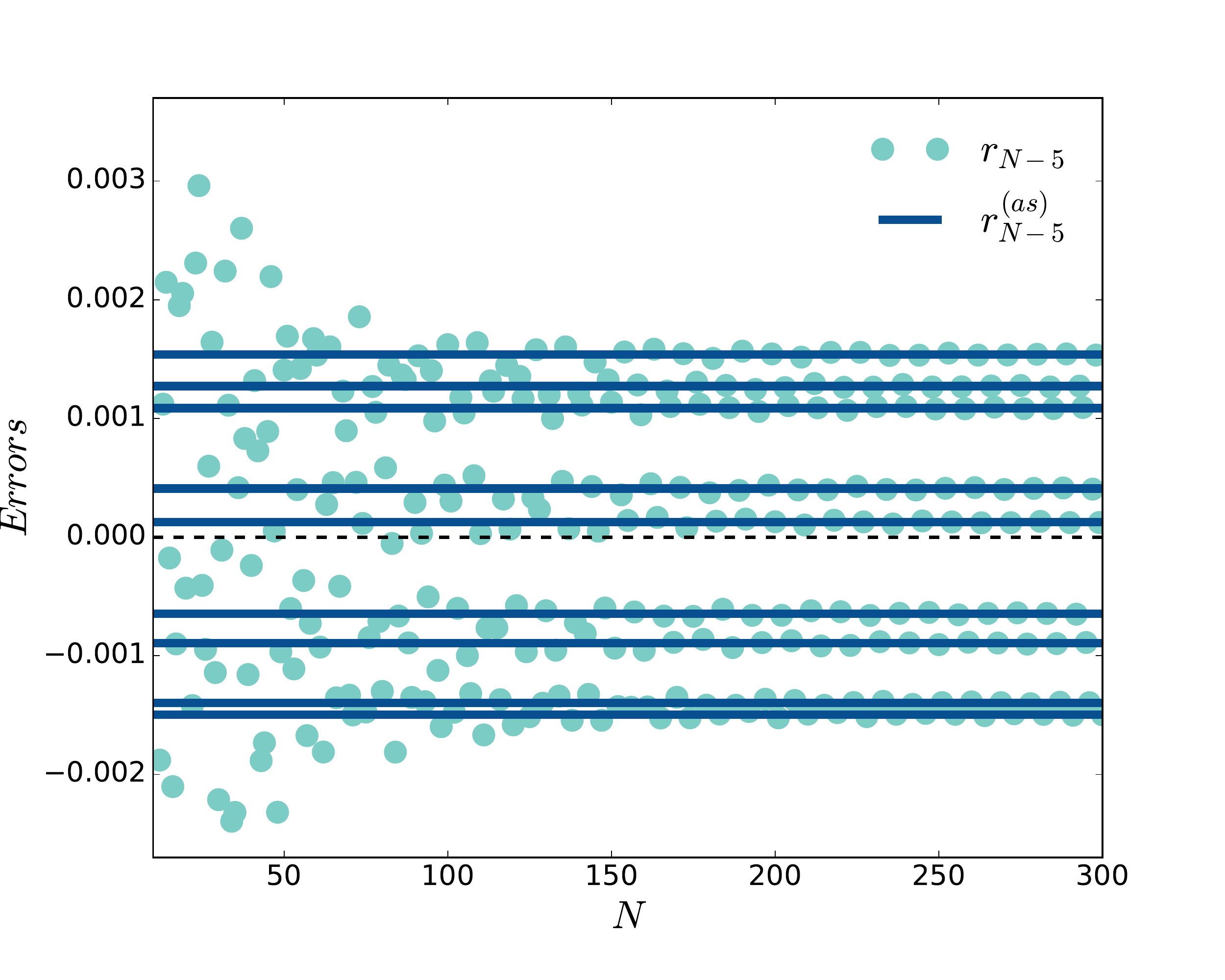}
\label{fig:omega=2_9_i=N-5}
\caption{Reconstruction errors $r_{N-5}$ and their limit values $r^{(as)}_{N-5}$
 for $\omega = 2/9$, $L = \lfloor 0.35 N\rfloor$, $a = 1.05$, $\delta = 0.1$, $\varphi = 0$ and $10 \leq N \leq 300$.}
\label{fig:omega_2_9_i_last}
\end{figure}
The convergence \eqref{eq:lim_r_j} and Proposition \ref{prop:limpoints_rat} are illustrated by Fig. \ref{fig:omega_2_9_i_last}.
To investigate the case of the irrational $\omega$ we use the following famous  equidistribution theorem  going
 back to P. Bohl \cite{Bohl_1909} and W. Sierpinski \cite{Sierpinski_1910}.
\begin{theorem}
\label{theor:BSW}
If $\alpha\in (0,1)$ is irrational, then the sequence  $z_n=\{n\alpha\}$ is uniformly distributed on $[0,1]$ in the sense that for any
$0\leq a< b\leq 1$
\begin{gather}
\label{eq:equidistr}
\frac{1}{n} \sum_{i=1}^{n} {\bf 1}_{[a,b)}(z_i)\rightarrow b-a, \quad n\rightarrow \infty.
\end{gather}

\end{theorem}
%

\begin{proposition}
\label{prop:r_irrat}
Let $\omega\in (0,1/2)$ be the irrational number and assume $\ell\geq 0$ to be fixed. Then for any $(a,b)\subset [D(-\ell), D(\ell)]$ and any $\delta\neq 0$
\begin{gather*}
\frac{1}{N} \sum_{n=\ell+1}^{N+\ell} {\bf 1}_{(a,b)}(r_{n-1-\ell}/\delta)\rightarrow \int_a^b \frac{1}{\pi \sqrt{D^2(\ell)-u^2}}\ du
\quad \text{as} \ \ N\rightarrow \infty.
\end{gather*}
\end{proposition}
\begin{proof}
In terms of the weak convergence
of distributions (see \cite{Billingsley_1999} for the whole theory), the convergence \eqref{eq:equidistr} means that
$\mathcal{P}_n\Rightarrow \mathrm{U}(0,1)$ as $n\rightarrow \infty$, where $\mathcal{P}_n$ stands for the uniform
 distribution on the set $\{z_{1},\ldots, z_{n}\}$, $\mathrm{U}(0,1)$ is the uniform distribution on $[0,1]$, and $``\Rightarrow''$ is the sign of the weak convergence.

 Now let us consider the sequence $\{\beta_n\}_{n\geq 1}$ of random variables defined on a certain probability space $(\Omega, \mathcal{F}, \mathrm{P})$
 and such that $\mathcal{L}(\beta_n)=\mathcal{P}_n$ for any $n$. (Note that here and further $\mathcal{L}(\beta)$ stands for the distribution of the
 random variable $\beta$.) Then \eqref{eq:equidistr} can be rewritten as $\mathcal{L}(\beta_n)\Rightarrow \mathcal{L}(\upsilon)$ as $n\rightarrow \infty$,
  where  $\upsilon\in \mathrm{U}(0,1)$.

 According to the Mapping Theorem \cite[theor. 2.7]{Billingsley_1999}, if $n\rightarrow \infty$ then
 \begin{gather*}
\mathcal{L}\big(h(\beta_n)\big)\Rightarrow \mathcal{L}\big(h(\upsilon)\big)
 \end{gather*}
 with
$
 h(z)=D(\ell)\sin\big(2\pi z+\varphi_1(\ell)\big).
$
 Standard calculations  show that the random variable $\eta=h(\upsilon)$ has the probability density
 \begin{gather}
 \label{eq:theor_dens}
p_{\eta}(z) = \frac{\mathbf{1}_{[-D(\ell), D(\ell)]}(z)}{\pi\sqrt{D(\ell)^2 - z^2}}\ ,
\end{gather}
where $\mathbf{1}_{A}(x)$ stands for the indicator function of the set $A$.

Since $\sin(2\pi j\omega+\phi)=\sin\big(2\pi \{j\omega\}+\phi\big)$ for any integer $j\geq 1$,
this means that for any $a<b$
\begin{gather*}
\frac{1}{N} \sum_{n=\ell+1}^{N+\ell} {\bf 1}_{(a,b)}\big(F_{n}(\ell)\big)\rightarrow \int_a^b p_\eta(u) du,
\end{gather*}
where $F_N(\ell)$ is defined in  \eqref{eq:FN} and $N\rightarrow \infty$.
In view of \eqref{eq:last_indices1} the assertion is proved.
\end{proof}

The result of Proposition \ref{prop:r_irrat} is illustrated by Fig. \ref{fig:omega_sqrt2_6_density}.
\begin{figure}[h]
\centering
\includegraphics[width=.7\linewidth]{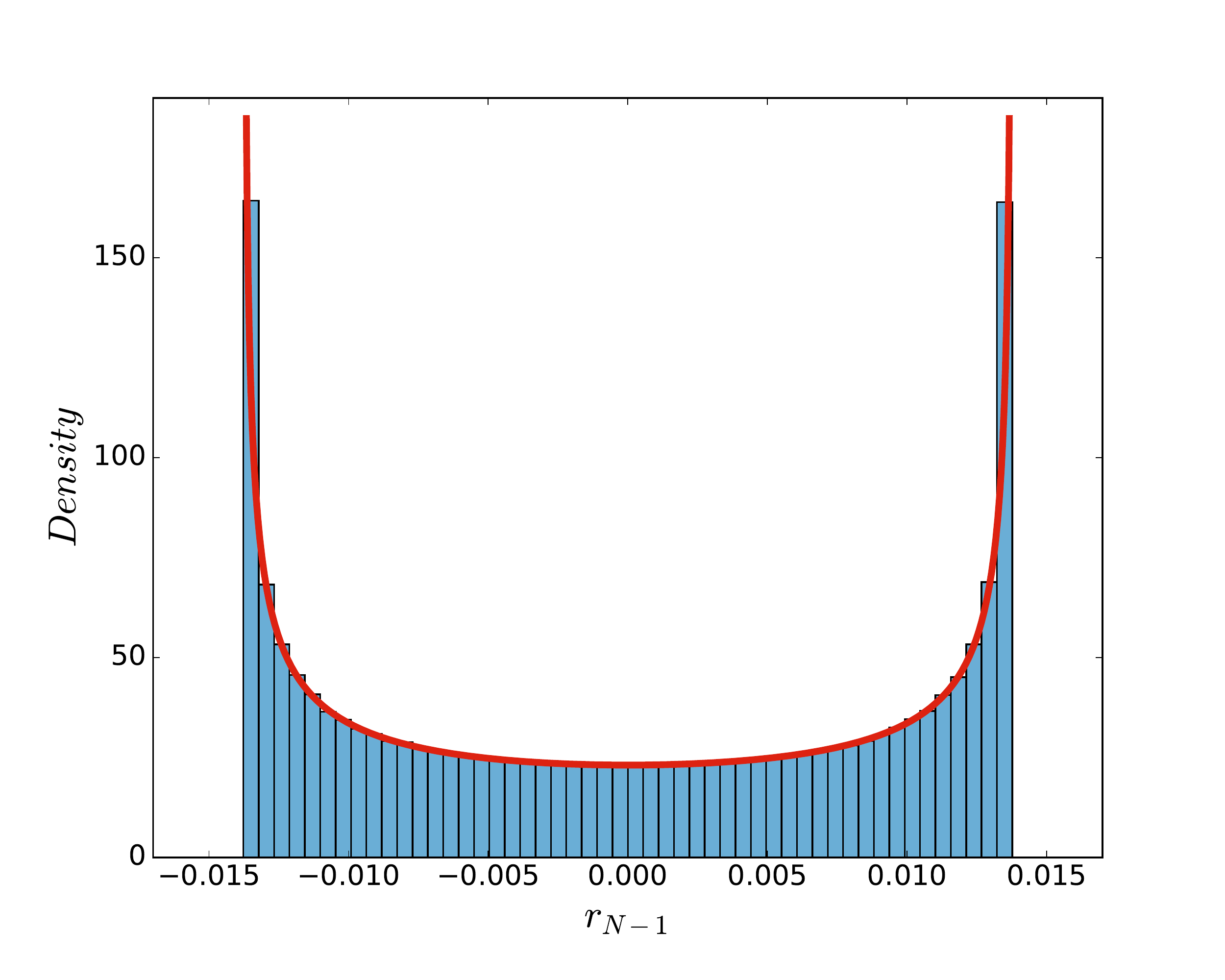}
\caption{Reconstruction errors $r_{N - 1}$ for $\omega = \sqrt{2}/6$, $L = \lfloor 0.35 N\rfloor$, $a = 1.05$, $\delta = 0.1$, $\varphi = 0$ and $10^3 \leq N \leq 10^6$. The histogram and the theoretical density  \eqref{eq:theor_dens}.}
\label{fig:omega_sqrt2_6_density}
\end{figure}

\begin{remark}
Propositions \ref{prop:limpoints_rat} and \ref{prop:r_irrat} show that for fixed $\ell$ and any $\omega\in (0,1/2)$ the reconstruction error $r_{N - 1 - \ell}$ does not converge to any limit value as $L/N\rightarrow \alpha\in (0,1)$. The case $\omega = 1/2$ can be studied in the same manner and gives the analogous result, while the exponential signal and the constant noise (i.e., the case $\omega=0$) is already checked in \cite{Nekrutkin10}.
\end{remark}


\section{Reconstruction errors for the signal $x_n=a^{nT/N}$}
\label{sect:rec_ant}
Now we deal with the discretization of the exponential signal, described in the Introduction. More precisely, we consider the constant $T>0$ and the
triangular  array of the series
\begin{gather}
\label{eq:discr_sign_noise}
f_n=f_n^{(N)}=a^{nT/N}+\delta\cos(\xi n+\varphi), \quad n=0,\ldots, N-1, \quad N=1,2,\ldots
\end{gather}
under the assumption that $N\rightarrow \infty$ and $L\sim \alpha N$ with $\alpha \in (0,1)$.
As in the previous section, we suppose that $a>1$, $\xi\in (0,\pi)$ and $\varphi\in [0,2\pi)$.

Of course, all formulas borrowed from \cite{Nekrutkin10} are valid here. Moreover, we can use all general formulas of Section \ref{sect:rec_an} if we put
$a^{jT/N}$ instead of $a^j$.
For example, now  we put
$W_j=\big(1,a^{T/N},\ldots,a^{(j-1)T/N}\big)^{\rm T}$ instead of denotation $W_j=\big(1,a,\ldots,a^{j-1}\big)^{\rm T}$ that was used in  Section \ref{sect:rec_an}.

In particular, since ${\rm rank}\, {\bf H}=1$, then the unique positive eigenvalue $\mu$ of the matrix ${\bf H} {\bf H}^{\rm T}$ has the form
\begin{gather}
\label{eq:exp_muN}
\mu=\big\|W_L\big\|^2\,\big\|W_K\big\|^2=
\displaystyle{\frac{(a^{2LT/N}-1) (a^{2KT/N}-1)}{(a^{2T/N}-1)^2}}\ .
\end{gather}

To investigate the discretization case we apply two general inequalities demonstrated in \cite{Nekrutkin10}. Here we put these statements in the form
adapted to
our problem.
Denote
\begin{gather}
\label{eq:Bdelta}
{\bf B}(\delta)=\delta\big({\bf H} {\bf E}^{\rm T}+{\bf E}{\bf H}^{\rm T}\big)+\delta^2{\bf E} {\bf E}^{\rm T}=
\delta {\bf A}^{(1)} + \delta^2{\bf A}^{(2)}
\end{gather}
 and let $\mu$ be defined by \eqref{eq:exp_muN}.

\begin{theorem}{\rm (\cite[theor. 2.3]{Nekrutkin10})}.
\label{theor:gen_upper}
If $\delta_0>0$ and $\big\|{\bf B}(\delta)\big\|/\mu<1/4$ for any $\delta\in (-\delta_0,\delta_0)$, then
\begin{gather*}
\big\|{\bf P}_0^{\perp}(\delta)-{\bf P}_0^{\perp}\big\|\leq
4C\ \frac{\|{\bf S}_0{\bf B}(\delta){\bf P}_0\|}{1-4\|{\bf B}(\delta)\|/\mu}\quad \text{with} \ \ C=e^{1/6}/\sqrt\pi.
\end{gather*}
\end{theorem}

\begin{theorem}{\rm (\cite[theor. 2.5]{Nekrutkin10})}
\label{theor:L_delta_proj}
Put $B(\delta)=|\delta|\, \|{\bf A}^{(1)}\|+\delta^2 \|{\bf A}^{(2)}\|$ and
assume that $\delta_0>0$, $B(\delta_0)=\mu/4$ and $|\delta|< \delta_0$. Denote ${\bf A}_0^{(2)} = {\bf P}_0{\bf A}^{(2)}{\bf P}_0$,
then $\|\delta{\bf A}_0^{(2)} \|<1$ and the matrix ${\bf I} - \delta{\bf A}_0^{(2)}$ is invertible.

Besides, under denotation ${\bf L}(\delta) = {\bf L}_1(\delta) + {\bf L}_1^{\rm T}(\delta)$ with
\begin{gather}
\label{eq:L1}
{\bf L}_1(\delta) = \frac{{\bf P}_{0}^{\perp}{\bf B}(\delta){\bf P}_0}{\mu} \Big({\bf I} - \delta{\bf A}_0^{(2)}/\mu\Big)^{-1},
\end{gather}
 the inequality
\begin{gather}
\big\|{\bf P}_0^{\perp}(\delta)-{\bf P}_0^{\perp}-{\bf L}(\delta)\big\|
\leq 16\,C
\frac{\|{\bf S}_0{\bf B}(\delta)\|\|{\bf S}_0{\bf B}(\delta){\bf P}_0\|}{1-4\|{\bf B}(\delta)\|/{\mu}}
\leq 16\,C
\frac{\|{\bf S}_0{\bf B}(\delta)\|^2}
{1-4\|{\bf B}(\delta)\|/{\mu}}
\label{eq:L_delta_proj_mainterm}
\end{gather}
is valid with the same $C$ as in Theorem  \ref{theor:gen_upper}.
\end{theorem}


\subsection{The convergence of $\big\|{\bf P}_0^{\perp}(\delta)-{\bf P}_0^{\perp}\big\|$}
We start with the norm of ${\bf B}(\delta)$.
\begin{lemma}
\label{lem:b_delta_t0}
Assume that $L\sim \alpha N$ with $\alpha\in (0,1)$.
Then there exist $\delta_0 > 0$, $N_0$ and $C$ such that $C\delta_0^2<1/4$ and
\begin{gather*}
 \|{\bf B}(\delta)\|/\mu \leq B(\delta)/\mu
 \leq C\delta^2
\end{gather*}
for  any $\delta$ with $|\delta|\leq \delta_0$ and $N>N_0$.
\end{lemma}
\begin{proof}
First of all,
\begin{gather*}
{\bf E}{\bf E}^{\rm T} = \frac{K}{2}
\begin{pmatrix}
1 & \cos \xi & \ldots & \cos((L - 1)\xi)\\
\ldots & \ldots & \ldots & \ldots\\
\cos((L - 1)\xi) & \cos((L - 2)\xi) & \ldots & 1
\end{pmatrix} +\\
+ \frac{\sin(K\xi)}{2\sin\xi}
\begin{pmatrix}
\cos((K - 1)\xi + 2\varphi) & \cos(K\xi + 2\varphi) & \ldots & \cos((N - 1)\xi + 2\varphi)\\
\ldots & \ldots & \ldots & \ldots\\
\cos((N - 1)\xi + 2\varphi) & \cos(N\xi + 2\varphi) & \ldots & \cos((N + L - 2)\xi + 2\varphi)
\end{pmatrix}.
\end{gather*}

Since $\sin\xi > 0$ and  $\|\mathbf{A}\| \leq \sqrt{\ell k}\|\mathbf{A}\|_{\max}$ for any matrix $\mathbf{A}:\mathbb{R}^\ell \mapsto \mathbb{R}^k$, then
\begin{gather}
\label{eq:EE}
\|{\bf E}{\bf E}^{\rm T}\| \leq \frac{LK}{2} + \frac{|\sin(K\xi)|L}{2\sin\xi}
\sim \alpha(1-\alpha) N^2/2
\end{gather}
as $N\rightarrow \infty$.
Using the analogue of \eqref{tq:EWK} we see that
\begin{gather}
\|{\bf H}{\bf E}^{\rm T}+{\bf E}{\bf H}^{\rm T}\|\leq 2 \|{\bf E}W_K\|\,\|W_L\|\leq
2\sqrt{\sum_{\ell = 0}^{L - 1}\varPhi_K^2(a^{T/N}, \xi \ell + \varphi)\, \frac{a^{2LT/N} - 1}{a^{2T/N} - 1}}\ .
\label{eq:HE_discr}
\end{gather}

It can be checked that $|\varPhi_{K}(a^{T/N}, \psi)| \leq C$ with a certain constant $C=C(a,T,\alpha,\xi)$ that does not depend on $\psi$.
For the further use we denote
\begin{gather}
\label{eq:maxC}
C_1 = \max\big(C(a,T,\alpha,\xi), C(a,T,1-\alpha, \xi)\big).
\end{gather}
 Since
\begin{gather*}
 \frac{a^{2LT/N} - 1}{a^{2T/N} - 1}\sim \frac{a^{2\alpha T}-1}{2T\ln a}\ N
\end{gather*}
as $N\rightarrow \infty$,
then it follows from \eqref{eq:HE_discr} that
\begin{gather*}
\|{\bf B}(\delta)\| \leq B(\delta)\leq
\delta^2\, \frac{\alpha(1 - \alpha)}{2}\, N^2 + o(N^2).
\end{gather*}
In view of the asymptotic
\begin{gather}
\label{eq:mu_as}
\mu = \frac{(a^{2\alpha T} - 1)(a^{2(1 - \alpha)T} - 1)}{4T^2\ln^2 a}\, N^2 + o(N^2),
\end{gather}
the proof is complete.
\end{proof}
\begin{theorem}
\label{theor:diff_proj_discr}
Under the conditions of Lemma \ref{lem:b_delta_t0},
$\big\|\mathbf{P}_0^\perp(\delta)-\mathbf{P}_0^\perp\big\|=O(N^{-1})$
as $N\rightarrow \infty$.
 \end{theorem}
\begin{proof}
Due to Theorem \ref{theor:gen_upper} and Lemma \ref{lem:b_delta_t0}, all we need is to proof that
\begin{gather}
\label{eq:SOB}
\|{\bf S}_0{\bf B}(\delta)\| = O(1/N).
\end{gather}
By \eqref{eq:P0S0},
\begin{gather}
{\bf S}_0{\bf B}(\delta) =
\delta\,\frac{{\bf D}_1}{\|W_L\|^2\|W_K\|^2} + \delta\,\frac{{\bf D}_2}{\|W_L\|^4\|W_K\|^2} + \delta^2\,
\frac{{\bf D}_3}{\|W_L\|^4\|W_K\|^2}
\label{eq:SOD_discr1}
\end{gather}
with
\begin{gather}
\label{eq:D123}
{\bf D}_1=W_LW_K^{\rm T}{\bf E}^{\rm T},\quad {\bf D}_2=W_LW_L^{\rm T}{\bf E}W_KW_L^{\rm T},\ \ \text{and}\ \
{\bf D}_3=W_LW_L^{\rm T}{\bf E}{\bf E}^{\rm T}.
\end{gather}

Consider summands in the righthand side of \eqref{eq:SOD_discr1} separately.
First of all,
\begin{gather*}
{\bf D}_3{\bf D}_3^{\rm T} = W_LW_L^{\rm T}{\bf E}{\bf E}^{\rm T}{\bf E}{\bf E}^{\rm T}W_LW_L^{\rm T}
= W_LW_L^{\rm T}\sum_{i = 0}^{L - 1}\left(\sum_{j = 0}^{L - 1}a^{jT/N}\varPsi_K(\varphi, i, j)\right)^2,
\end{gather*}
where
\begin{gather}
\label{eq:psi_l_eq}
\varPsi_M(\psi, k, \ell) = \sum_{j = 0}^{M - 1}\cos(\xi (j + k) + \psi)\cos(\xi (j + \ell) + \psi).
\end{gather}

Since
\begin{gather}
\sum_{j = 0}^{L - 1}a^{jT/N}\varPsi_K(\varphi, i, j) =
\frac{K}{2}\, \varPhi_L(a^{T/N}, -i\xi) +
 \frac{\sin(K\xi)}{2\sin\xi}\, \varPhi_L(a^{T/N}, i\xi + (K - 1)\xi + 2\varphi),
 \label{eq:sumvarphi}
\end{gather}
we get
\begin{gather*}
\|{\bf D}_3\|^2 \!\leq\!\frac{a^{2TL/N}\! -\!1}{a^{2T/N} - 1}\left(C_1^2\,\frac{K^2L}{4} + O(LK)\!\right)
=\frac{a^{2\alpha T} - 1}{2T \ln a}\, C_1^2\,\frac{(1 - \alpha)^2\alpha}{4}\, N^4 + o(N^4)
\end{gather*}
with $C_1$ defined in \eqref{eq:maxC}.
Thus
$\|{\bf D}_3\|\, \big(\|W_L\|^4\|W_K\|^2\big)^{-1}=O(1/N)$.
In the same manner,
\begin{gather*}
{\bf D}_2 =
W_LW_L^{\rm T}\sum_{i = 0}^{L - 1} a^{iT/N} \varPhi_K(a^{T/N}, i\xi + \varphi)
\end{gather*}
and
\begin{gather*}
\|{\bf D}_2\| \leq
C_1\sqrt{\frac{a^{2\alpha T} - 1}{2}}\, \frac{a^{\alpha T} - 1}{(T\ln a)^{3/2}}\,  N^{3/2} + o(N^{3/2}).
\end{gather*}
Therefore, $\|{\bf D}_2\|\, \big(\|W_L\|^4\|W_K\|^2\big)^{-1}=O(1/N)$.
Lastly,
\begin{gather*}
{\bf D}_1{\bf D}_1^{\rm T}  = W_LW_K^{\rm T}{\bf E}^{\rm T}{\bf E}W_KW_L^{\rm T}
= W_LW_L^{\rm T} \sum_{i = 0}^{L - 1}\varPhi_K^2(a^{T/N}, i\xi + \varphi).
\end{gather*}
Since
$
\sum_{i = 0}^{L - 1}\varPhi_K^2(a^{T/N}, i\xi + \varphi) \leq L C_1^2,
$
then
\begin{gather*}
\|{\bf D}_1\|^2 \leq \frac{a^{2T_0L/N}\! - \!1}{a^{2T_0/N} \!- \!1}\, C_1^2 L
=
\alpha\, \frac{a^{2\alpha T} - 1}{2T \ln a}C_1^2\,N^2\! + o(N^2),
\end{gather*}
$\|{\bf D}_1\|\,\big(\|W_L\|^2\|W_K\|^2\big)^{-1}=O(1/N)$ and the proof is complete.
%
\end{proof}

\subsection{Reconstruction errors}

To investigate  the reconstruction errors we use the same idea as in Section \ref{sect:rec_an}
but deal with the inequality \eqref{eq:L_delta_proj_mainterm} instead of
\eqref{eq:main_term V} and use the
 expression
\begin{gather}
\label{eq:deltaH_pres2}
\Delta_{\delta}({\bf H}) =
\big({\bf P}_0^\perp(\delta) - {\bf P}_0^\perp- {\bf L}(\delta)\big){\bf H}(\delta) + \delta {\bf P}_0^\perp{\bf E}+
{\bf L}(\delta){\bf H}+\delta{\bf L}(\delta)\mathbf{E}\ .
\end{gather}
instead of \eqref{eq:deltaH_pres1}.
For this goal, we need the following supplementary assertions.

\begin{lemma}
\label{lem:Zmax1}
Denote ${\bf Z} = \delta{\bf A}_0^{(2)}/\mu=\delta \mathbf{P}_0{\bf EE}^{\rm T}\mathbf{P}_0/\mu$.
Then
there exists a  constant $C_2$ such that
\begin{gather}
\|{\bf Z}\|_{\max} \leq {|\delta| C_2}/{N}
\label{eq:Zmax1}
\end{gather}
 with  $\|{\bf Z}\|_{\max}=\max_{m,\ell} \big|{\bf Z}[m,\ell]\big|$.
 \end{lemma}

\begin{proof}
First of all,
\begin{gather*}
{\bf A}_0^{(2)}[m, \ell]\! = \!{\bf E}{\bf E}^{\rm T}[m, \ell]\! -\! \frac{{\bf E}{\bf E}^{\rm T}W_LW_L^{\rm T}[m, \ell]}{\|W_L\|^2}\! -\!
 \frac{W_LW_L^{\rm T}{\bf E}{\bf E}^{\rm T}[m, \ell]}{\|W_L\|^2}\! +\!
\frac{W_LW_L^{\rm T}{\bf E}{\bf E}^{\rm T}W_LW_L^{\rm T}[m, \ell]}{\|W_L\|^4}\,.
\end{gather*}
Note that ${\bf E}{\bf E}^{\rm T}[m, \ell]=\varPsi_K(\varphi, m, \ell)$, where $\varPsi_M(\psi, k, \ell)$ is defined in \eqref{eq:psi_l_eq}. Analogously,
\begin{gather*}
{\bf E}{\bf E}^{\rm T}W_LW_L^{\rm T}[m, \ell]=
a^{ \ell T/N}\sum_{j = 0}^{L - 1}a^{jT/N}\varPsi_K(\varphi, m, j),
\\
W_LW_L^{\rm T}{\bf E}{\bf E}^{\rm T}[m, \ell]=
a^{m T/N}\sum_{j = 0}^{L - 1}a^{jT/N}\varPsi_K(\varphi, j, \ell), \quad  \text{and}\\
W_LW_L^{\rm T}{\bf E}{\bf E}^{\rm T}W_LW_L^{\rm T}[m, \ell]=
a^{(\ell+m)T/N}\sum_{k ,j= 0}^{L - 1}a^{(j  + k)T/N}\varPsi_K(\varphi, k, j).
\end{gather*}
In view of \eqref{eq:sumvarphi},
\begin{gather*}
\frac{{\bf E}{\bf E}^{\rm T}W_LW_L^{\rm T}[m, \ell]}{\|W_L\|^2}\, +
 \frac{W_LW_L^{\rm T}{\bf E}{\bf E}^{\rm T}[m, \ell]}{\|W_L\|^2}\, +
 \frac{W_LW_L^{\rm T}{\bf E}{\bf E}^{\rm T}W_LW_L^{\rm T}[m, \ell]}{\|W_L\|^4}\, = O(1).
\end{gather*}
Since
$
\varPsi_K(\varphi, m, \ell)={K} \cos(\xi(m-\ell))/2+O(1)
$
as $K\rightarrow \infty$,
then
\begin{gather*}
{\bf A}_0^{(2)}[m, \ell]
=\frac{\cos((m - \ell)\xi)}{2}\, K + o(1)
\end{gather*}
uniformly in  $0 \leq m, \ell \leq L - 1$.
Therefore, see \eqref{eq:mu_as},
\begin{gather*}
{\bf Z}[m, \ell] = \delta \cos((m - \ell)\xi)\  \frac{2(1 - \alpha)T^2\ln^2 a}{(a^{2\alpha T_0} - 1)(a^{2(1 - \alpha)T_0} - 1)}\ \frac{1}{N} + o(1/N)
\end{gather*}
and the proof is complete.
\end{proof}
\begin{remark}
As the consequence of the inequality \eqref{eq:Zmax1} we get  that for any $n\geq 1$
 \begin{gather*}
  \|{\bf Z}^n\|_{\max}\leq {|\delta|^n C_2^n}/{N}.
 \end{gather*}
 Since
$
 \|{\bf Z}^n\|_{\max}\leq L \|{\bf Z}^{n-1}\|_{\max} \|{\bf Z}\|_{\max},
 $
this fact can be proved with the help of a simple induction.
Therefore, if $|\delta|C_2<1$, then
\begin{gather}
\label{eq:sumZ_max}
\Big\|\sum_{n\geq 1}\mathbf{Z}^n\Big\|_{\max}\leq \sum_{n\geq 1}\left\| \mathbf{Z}^n\right\|_{\max}\leq \frac{|\delta| C_2}{1-|\delta| C_2}\, \frac{1}{N}\ .
\end{gather}

\end{remark}

\begin{lemma}
\label{lem:addit_norms}
If the series $f_n=f_n^{(N)}$ is defined by \eqref{eq:discr_sign_noise}, $N\rightarrow \infty$ and $L\sim \alpha N$ with $\alpha\in (0,1)$, then
1) $\|{\bf B}(\delta){\bf H}\|_{\max} = O(N)$;
2) $\|{\bf S}_0{\bf B}(\delta)\|_{\max} = O(1/N^2)$;
3) $\|{\bf B}(\delta){\bf S}_0{\bf E}\|_{\max} = O(1/N^2)$ and
4) $\|{\bf P}_0^{\perp} {\bf E}\|_{\max}  = O(1/N)$.
\end{lemma}
\begin{proof}
1) The matrix ${\bf B}(\delta){\bf H}$ can be rewritten as follows:
\begin{gather*}
{\bf B}(\delta){\bf H} =
\delta \|W_L\|^2 {\bf J}_1 + \delta {\bf J}_2 + \delta^2 {\bf J}_3
\end{gather*}
with
$
{\bf J}_1= {\bf E}W_KW_K^{\rm T}$, ${\bf J}_2 =W_LW_K^{\rm T}{\bf E}^{\rm T}W_LW_K^{\rm T}$ {and}
${\bf J}_3={\bf E}{\bf E}^{\rm T}W_LW_K^{\rm T}$.

Applying the equalities \eqref{eq:upsilon} and \eqref{eq:psi_l_eq}, \eqref{eq:sumvarphi}
we get that
\begin{gather*}
\|{\bf J}_1\|_{\max} = \max_{k < L,\, \ell < K}\, \big|\varPhi_K(a^{T/N}, k\xi + \varphi)a^{\ell T/N}\big| \leq C_1\,a^{(1 - \alpha)T}
 + o(1),\\
\|{\bf J}_2\|_{\max} = \big|\varUpsilon_{L,\, K}(a^{T/N}, \varphi)\big|\max_{k < L; \ell < K} a^{(k + \ell)T/N} \leq
C_1\, \frac{a^{T}(a^{\alpha T} - 1)}{T\ln a}\, N + o(N),\\
\|{\bf J}_3\|_{\max} = \max_{k < L,\, \ell < K}
\Bigg|\sum_{j = 0}^{L - 1}a^j\varPsi_K(\varphi, k, j)a^{\ell T/N}\Bigg| \leq C_1\,\frac{(1 - \alpha)a^{(1 - \alpha)T} }{2}\,
 N + o(N)
\end{gather*}
with $C_1$ defined in \eqref{eq:maxC}.
Thus the first assertion is proved.

2) The expression for ${\bf S}_0{\bf B}(\delta)$ is presented in \eqref{eq:SOD_discr1}, \eqref{eq:D123}. It can be checked that
\begin{gather*}
\|{\bf D}_1\|_{\max} = \max_{k < L; \ell < L}\, \big|\varPhi_K(a^{T/N}, k\xi + \varphi)a^{\ell T/N}\big|
\leq C_1\, a^{\alpha T} + o(1),\\
\|{\bf D}_2\|_{\max} = \big|\varUpsilon_{L, K}(a^{T/N}\!, \varphi)\big|\max_{k, \ell < L} a^{(k + \ell)T/N} \leq C_1 \frac{a^{2\alpha T}(a^{\alpha T} - 1)}{T\ln a}\,  N + o(N),\\
\|{\bf D}_3\|_{\max} = \max_{k, \ell < L}
\Bigg|\sum_{j = 0}^{L - 1}a^j\varPsi_K(\varphi, k, j)a^{\ell T/N}\Bigg| \leq C_1\,\frac{(1 - \alpha)a^{\alpha T} }{2}\, N + o(N).
\end{gather*}

Applying \eqref{eq:SOD_discr1} we see that
 $\|{\bf S}_0{\bf B}(\delta)\|_{\max} = O(1/N^2)$.

3)
In the same manner,
\begin{gather*}
{\bf B}(\delta){\bf S}_0{\bf E} = \delta {\bf E}{\bf H}^{\rm T}{\bf S}_0{\bf E} + \delta {\bf H}{\bf E}{\bf S}_0{\bf E}
+ \delta^2{\bf E}{\bf E}^{\rm T}{\bf S}_0{\bf E} =\\
=\frac{1}{\|W_L\|^4\|W_K\|^2}\big(\delta\|W_L\|^2{\bf E}W_KW_L^{\rm T}{\bf E} +
\delta W_LW_K^{\rm T}{\bf E}^{\rm T}W_LW_L^{\rm T}{\bf E} + \delta^2{\bf E}{\bf E}^{\rm T}W_LW_L^{\rm T}{\bf E}\big)
\end{gather*}
with
\begin{gather*}
\big\|{\bf E}W_KW_L^{\rm T}{\bf E}\big\|_{\max} = \max_{k < L; \ell < K}\, \big|\varPhi_K(a^{T/N}, k\xi + \varphi)\varPhi_L(a^{T/N}, \ell\xi + \varphi)\big| \leq
 C_1^2 + o(1),\\
\big\|W_LW_K^{\rm T}{\bf E}^{\rm T}W_LW_L^{\rm T}{\bf E}\big\|_{\max} = \big|\varUpsilon_{L, K}(a^{T/N}\!, \varphi)\big|
\max_{k < L; \ell < K} \big|a^{kT/N}\varPhi_L(a^{T/N}, \ell\xi + \varphi)\big| \leq\\
\leq C_1^2\,\frac{a^{\alpha T}(a^{\alpha T} - 1)}{T\ln a}\, N + o(N), \quad \text{and}\\
\big\|{\bf E}{\bf E}^{\rm T}W_LW_L^{\rm T}{\bf E}\big\|_{\max} = \max_{k < L; \ell < K} \Bigg|\sum_{j = 0}^{L - 1}a^j\varPsi_K(\varphi, k, j)\varPhi_L(a^{T/N}, \ell\xi + \varphi)\Bigg|
\leq N C_1^2\,{(1 - \alpha)}/{2} + o(N).
\end{gather*}
Therefore,
$\|{\bf B}(\delta){\bf S}_0{\bf E}\|_{\max} = O(1/N^2)$.

4)
Since
${\bf P}_0^{\perp} {\bf E} = W_LW_L^{\rm T}{\bf E}/\|W_L\|^2,$
then
\begin{gather*}
\|{\bf P}_0^{\perp} {\bf E}\|_{\max} = \frac{1}{\|W_L\|^2}\, \max_{k < K; \ell < L}\, \big|\varPhi_L(a^{T/N}, k\xi + \varphi)a^{\ell T/N}\big|
\leq C_1\,\frac{2T\ln a}{a^{2\alpha T} - 1}\, a^{\alpha T}\, \frac{1}{N} + o(1/N),
\end{gather*}
and the proof is complete.
\end{proof}

\begin{theorem}
\label{theor:main_discr}
Denote $r_j = r_j(N,\delta)$ the reconstruction error for the term $x_j = a^{jT/N}$ of the perturbed series
$f_j = x_j + \delta\cos(\xi j + \varphi)$ with $a > 1$ and $\xi\in (0,\pi)$, $\varphi\in [0,2\pi)$.

If $N \rightarrow \infty$ and $L = \alpha N + o(N)$ with $0 < \alpha < 1$,
then there exists $\delta^*>0$ such that
$r_j = O(1/N)$ uniformly in $0\leq j< N$ for any $\delta$ with $|\delta|<\delta^*$.
\end{theorem}
\begin{proof}
First of all, see Lemma \ref{lem:b_delta_t0}, the inequality \eqref{eq:L_delta_proj_mainterm} holds for any $\delta$
such that $|\delta|<\delta_0$.
Then, due to \eqref{eq:deltaH_pres2},
\begin{gather}
\label{eq:decomp}
{\bf P}_0^{\perp}(\delta){\bf H}(\delta) - {\bf P}_0^{\perp}{\bf H} = \big({\bf P}_0^{\perp}(\delta)-{\bf P}_0^{\perp}-{\bf L}(\delta)\big){\bf H}(\delta) + {\bf L}(\delta){\bf H} + \delta {\bf L}(\delta){\bf E} + \delta {\bf P}_0^{\perp} {\bf E}.
\end{gather}
It follows from
\eqref{eq:L_delta_proj_mainterm} and \eqref{eq:SOB} that
$
\|\big({\bf P}_0^{\perp}(\delta)-{\bf P}_0^{\perp}-{\bf L}(\delta)\big)\| = O(1/N^2)
$
for $|\delta|<\delta_0$, and therefore
\begin{gather*}
\left\|\big({\bf P}_0^{\perp}(\delta)-{\bf P}_0^{\perp}-{\bf L}(\delta)\big){\bf H}(\delta)\right\|\leq
\left\|\big({\bf P}_0^{\perp}(\delta)-{\bf P}_0^{\perp}-{\bf L}(\delta)\big)\right\|
\|\mathbf{H}(\delta)\|
 = O(1/N).
\end{gather*}
Thus we must check the tree last terms in the righthand side of \eqref{eq:decomp}.

Note that $\|{\bf P}_0^{\perp} {\bf E}\|_{\max} = O(1/N)$, see Lemma \ref{lem:addit_norms}.
Let us consider operators ${\bf L}(\delta){\bf H}$ and ${\bf L}(\delta){\bf E}$.
As in Lemma \ref{lem:Zmax1}, put ${\bf Z} = \delta{\bf A}_0^{(2)}/\mu$.
Since ${\bf P}_0{\bf H}=\mathbf{0}$, then ${\bf Z}{\bf H}=\mathbf{0}$ and
\begin{gather*}
{\bf L}_1(\delta){\bf H}={\bf S}_0{\bf B}(\delta){\bf P}_0\big({\bf I} - {\bf Z}\big)^{-1}{\bf H}=
{\bf S}_0{\bf B}(\delta){\bf P}_0\sum_{m\geq 0} \mathbf{Z}^m {\bf H}={\bf 0}.
\end{gather*}
Therefore,
\begin{gather*}
{\bf L}(\delta){\bf H}=
 {\bf L}^{\rm T}_1(\delta){\bf H}=\Big({\bf I} - {\bf Z}\Big)^{-1}\frac{{\bf P}_{0}{\bf B}(\delta){\bf H}}{\mu}=
\frac{1}{\mu}\left({\bf I}+\sum_{m\geq 1}\mathbf{Z}^m\right){\bf P}_{0}{\bf B}(\delta){\bf H}=
\\
\frac{1}{\mu}{\bf P}_{0}{\bf B}(\delta){\bf H}+
\frac{1}{\mu}\sum_{m\geq 1}
\mathbf{Z}^m{\bf P}_{0}{\bf B}(\delta){\bf H}=
\frac{1}{\mu}{\bf P}_{0}{\bf B}(\delta){\bf H}+
\frac{1}{\mu}\left(\sum_{m\geq 1}\mathbf{Z}^m\right){\bf B}(\delta){\bf H}.
\end{gather*}

Since $\mathbf{P}_0^\perp=W_LW_L^{\rm T}/\|W_L\|^2$ and
$
\|W_LW_L^{\rm T}\|_{\max}
\leq a^{2T}$, then $\|\mathbf{P}_0^\perp\|_{\max}=O(1/N)$
and, in view of Lemma \ref{lem:addit_norms},
\begin{gather*}
\big\|{\bf P}_{0}{\bf B}(\delta){\bf H}/\mu\big\|_{\max}=
\left\|\big(\mathbf{I}-\mathbf{P}_0^\perp\big){\bf B}(\delta){\bf H}/\mu\right\|_{\max}\leq
\\
\frac{\left\|{\bf B}(\delta){\bf H}\right\|_{\max}}{{\mu}}+
\frac{\left\|\mathbf{P}_0^\perp{\bf B}(\delta){\bf H}\right\|_{\max}}{\mu}\leq
\frac{\left\|{\bf B}(\delta){\bf H}\right\|_{\max}}{{\mu}}+
{L\left\|\mathbf{P}_0^\perp\right\|_{\max}}\ \frac{\left\|{\bf B}(\delta){\bf H}\right\|_{\max}}{{\mu}}=O(1/N).
\end{gather*}
Besides, if additionally $|\delta| C_2<1$, then,  due to \eqref{eq:mu_as} and \eqref{eq:sumZ_max},
\begin{gather*}
\Big\|\frac{1}{\mu}\Big(\sum_{m\geq 1}\mathbf{Z}^m\Big){\bf B}(\delta){\bf H}\Big\|_{\max}
\leq
\frac{L}{\mu}\Big\|\sum_{m\geq 1}\mathbf{Z}^m\Big\|_{\max}\ \Big\|{\bf B}(\delta){\bf H}\Big\|_{\max}\leq
\\
\frac{L}{\mu}\ \frac{|\delta| C_2}{1-|\delta| C_2}\, \frac{1}{N}\ \Big\|{\bf B}(\delta){\bf H}\Big\|_{\max}=O(1/N).
\end{gather*}
As the result, $\|{\bf L}(\delta){\bf H}\|_{\max}=O(1/N)$.


By definition, ${\bf L}(\delta){\bf E} = {\bf L}_1(\delta){\bf E}+{\bf L}^{\rm T}_1(\delta){\bf E}$ with
\begin{gather*}
{\bf L}_1(\delta){\bf E}=
{\bf S}_0{\bf B}(\delta){\bf P}_0 \big({\bf I} -{\bf Z})^{-1}{\bf E} \quad \text{and}\ \
{\bf L}^{\rm T}_1(\delta){\bf E}=\big({\bf I} - {\bf Z})^{-1}{\bf P}_0{\bf B}(\delta){\bf S}_0{\bf E}.
\end{gather*}

The equality $\|{\bf L}^{\rm T}_1(\delta){\bf E}\|_{\max}=O(1/N^{2})$ can be demonstrated
in the same manner as $\|{\bf L}^{\rm T}_1(\delta){\bf H}\|_{\max} = O(1/N)$, with the help of the equality
$\|{\bf B}(\delta){\bf S}_0{\bf E}\|_{\max} = O(1/N^2)$
(see Lemma~\ref{lem:addit_norms}).

Now note that
\begin{gather*}
 {\bf L}_1(\delta){\bf E}=
 {\bf S}_0{\bf B}(\delta){\bf P}_0 {\bf E}+{\bf S}_0{\bf B}(\delta){\bf P}_0\left(\sum_{m\geq 1}\mathbf{Z}^m\right){\bf E}=
  {\bf S}_0{\bf B}(\delta){\bf P}_0{\bf E}+{\bf S}_0{\bf B}(\delta)\left(\sum_{m\geq 1}\mathbf{Z}^m\right){\bf E}.
\end{gather*}
Taking into account  \eqref{eq:sumZ_max} and equalities $\left\|{\bf E}\right\|_{\max}=O(1)$,
$\|{\bf S}_0{\bf B}(\delta)\|_{\max} = O(1/N^2)$ we get
\begin{gather*}
\bigg\|{\bf S}_0{\bf B}(\delta)\Big(\sum_{m\geq 1}\mathbf{Z}^m\Big){\bf E}\bigg\|_{\max}\leq
L^2\left\|{\bf S}_0{\bf B}(\delta)\right\|_{\max}\Big\|\sum_{m\geq 1}\mathbf{Z}^m\Big\|_{\max}
\left\|{\bf E}\right\|_{\max}=O(1/N).
\end{gather*}

Lastly,
\begin{gather*}
\| {\bf S}_0{\bf B}(\delta){\bf P}_0{\bf E} \|_{\max}\leq
\| {\bf S}_0{\bf B}(\delta){\bf E} \|_{\max}+\| {\bf S}_0{\bf B}(\delta){\bf P}_0^{\perp}{\bf E} \|_{\max}\leq
\\
L \| {\bf S}_0{\bf B}(\delta)\|_{\max} \|{\bf E} \|_{\max} +L^2 \| {\bf S}_0{\bf B}(\delta)\|_{\max}\, \| {\bf P}_0^{\perp}\|_{\max}
\, \|{\bf E} \|_{\max}=O(1/N).
\end{gather*}

Therefore,
$\|{\bf L}(\delta){\bf E}\|_{\max} = O(1/N)$.
Finally, the uniform norm $\|\, \cdot\, \|_{\max}$ of each addend in the sum \eqref{eq:deltaH_pres2} has the order $O(1/N). $
Since $\|\mathcal{S}{\bf C}\|_{\max} \leq \|{\bf C}\|_{\max}$ for any matrix $\mathbf{C}$, the proof is complete.
\end{proof}






\end{document}